\documentclass[10pt,a4paper]{article}
\usepackage[margin=1in]{geometry}
\usepackage[T1]{fontenc}
\usepackage{lmodern}
\usepackage{microtype}
\usepackage[utf8]{inputenc}
\usepackage{amsthm,amsmath,amsfonts,amssymb}
\usepackage[authoryear]{natbib}
\setcitestyle{round}
\usepackage[hidelinks]{hyperref}
\usepackage{graphicx}
\usepackage{mathtools}
\usepackage{mathrsfs}
\usepackage{enumitem}
\usepackage{tikz}
\usepackage{comment}
\usetikzlibrary{arrows.meta,positioning,fit,calc}

\numberwithin{equation}{section}
\mathtoolsset{showonlyrefs=true}

\theoremstyle{plain}
\newtheorem{theorem}{Theorem}[section]

\newtheorem{proposition}[theorem]{Proposition}
\newtheorem{lemma}[theorem]{Lemma}

\theoremstyle{definition}
\newtheorem{definition}{Definition}[section]
\newtheorem{assumption}{Assumption}[section]
\newtheorem{example}{Example}[section]
\newtheorem{remark}{Remark}[section]

\setlist[enumerate]{leftmargin=2.7em,labelsep=.5em}
\newcommand{\T}{\mathcal T}
\newcommand{\E}{\mathbb E}
\newcommand{\Q}{\mathbb Q}
\newcommand{\PP}{\mathbb P}

\newcommand{\U}{\mathcal U}
\newcommand{\ind}{\mathbb I}
\newcommand{\te}{\tau_{\mathrm{exp}}}

\begin{document}
\title{Valuation Reveals Uncertainty}
\author{Jongjin Park\thanks{\raggedright Research Institute of Mathematics, Seoul National University, Seoul 08826, South Korea (\texttt{pjj4230@snu.ac.kr})}
\and David Criens\thanks{\raggedright Department of Mathematical Stochastics, University of Freiburg, 
 Freiburg 79104, Germany \\(\texttt{david.criens@stochastik.uni-freiburg.de})}
\and Hyungbin Park\thanks{\raggedright Department of Mathematical Sciences and Research Institute of Mathematics, Seoul National University, Seoul 08826, South Korea (\texttt{hyungbin@snu.ac.kr}, \texttt{hyungbin2015@gmail.com})}}

\date{}
\maketitle

\begin{abstract}
This paper studies the recovery of uncertainty from dynamic sublinear valuation rules. Robust valuation assigns each payoff its worst-case expected value over a family of plausible models, thereby inducing a dynamic sublinear valuation rule. While valuation rules are observable in practice, the underlying uncertainty structure remains latent. Under suitable conditions, we show that an uncertainty structure can be recovered from an observed valuation rule and provide an explicit recovery procedure. The construction characterizes the joint uncertainty in the discount factor and the law of the state process.
\end{abstract}

\par\smallskip

\par\medskip

\setcounter{secnumdepth}{3}
\setcounter{tocdepth}{2}
\tableofcontents

\section{Introduction}\label{sec:intro}

Valuation and uncertainty are fundamental objects in economics.
Robust valuation provides a widely used framework for studying their relationship.
Conventionally, uncertainty is specified a priori, and the corresponding valuation rule is derived from it.
This paper takes the converse perspective.
In practice, valuation is often observable through market prices, whereas the underlying uncertainty remains latent.
Our results show that, under suitable conditions, an uncertainty structure can be recovered from valuation.

This paper explores two types of continuous-time valuation rules. The first is a dynamic sublinear valuation rule, a time-homogeneous family of operators $\mathcal T=\{\mathcal T_t\}_{t\ge0}$ characterized by economic axioms without a prescribed probability measure or state process; see Definition~\ref{def:valuation}.
The second is a dynamic robust valuation rule, formulated over a family of probabilistic models that captures uncertainty in both state dynamics and discounting. Let \(X\) be an underlying state process with state space \(D\subset\mathbb R^d\).
For each \(x\in D\), consider a family \(\mathcal U_x\) of pairs \((A,\mathbb Q)\), where \(A\) is a cumulative discounting process and \(\mathbb Q\) specifies the law of \(X\) starting from \(x\).
Each pair represents a possible model, and \(\mathcal U_x\) captures the joint uncertainty in discounting and state dynamics at \(x\).
The robust valuation of a payoff function \(f\) over a horizon \(t\), starting from state \(x\), is defined by
\begin{equation}\label{eq:robust_valuation}
\T_t^{\U}f(x)
:=\sup_{(A,\Q)\in\U_x}\E^{\Q}\!\left[e^{-A_t}f(X_t)\ind_{\{t<\te\}}\right].
\end{equation}
where \(\tau_{\mathrm{exp}}\) is the explosion time of \(X\) in \eqref{eq:exit_times}.
The families \(\mathcal U=\{\mathcal U_x\}_{x\in D}\) and \(\mathcal T^{\mathcal U}=\{\mathcal T_t^{\mathcal U}\}_{t\ge0}\) are called an \emph{uncertainty structure} and its associated \emph{robust valuation rule}, respectively.

The main contribution of this paper is to show that, under suitable conditions, every dynamic sublinear valuation rule admits a representation as a robust valuation under an uncertainty structure.
More precisely, given such a rule $\mathcal T$, we provide an explicit procedure for recovering an uncertainty structure $\mathcal U$ such that
\[
\mathcal T
=
\mathcal T^{\mathcal U}.
\]
The construction is summarized in Figure~\ref{fig:linear_valuation_uncertainty}.
Our results provide a new perspective on the role of valuation in economics and finance. Valuation is not merely an outcome of uncertainty but also a source of information about the uncertainty structures that govern it. 
In this sense, our findings complement the Knightian view of uncertainty as inherently non-measurable \cite{knight1921risk}.
This approach provides a framework for studying latent uncertainty and opens new avenues for theoretical and empirical research on economic systems.

This paper is closely related to \cite{park2026global}.
The proofs of our main results rely substantially on that framework, and the present paper may therefore be regarded as a follow-up study.
In the earlier work, laws on an auxiliary path space with killing serve as technical objects for stochastic representation; they do not directly describe the distribution of the underlying economic state process.
The present paper gives these structures an economic interpretation in terms of joint uncertainty in discounting and state dynamics.

This interpretation requires two additional technical steps.
First, a cumulative discounting process and a law of the underlying state process are recovered from each killed law, making explicit the discounting uncertainty encoded through killing in \cite{park2026global}.
Second, our valuation rules act on $C_b(D)$, whereas the earlier framework is developed on $\mathrm{USC}_b(D)$.
The required recovery and continuity results are established here through additional probabilistic and analytic arguments.

Sublinear valuation rules connect robust finance with nonlinear expectations and nonlinear semigroup theory.
They arise in dual representations of superhedging prices under volatility uncertainty \cite{neufeld_nutz_13}.
Probabilistic representations of nonlinear expectations and convex functionals, including path-space constructions and stability under conditioning and pasting, are studied in \cite{bartl2020conditional,bartl_cheridito_kupper_19,criens2025representation,nutz2013constructing}.
A complementary approach starts from a prescribed family of probability measures, as in robust utility maximization \cite{bartl_kupper_neufeld_21}.
Connections between local characteristics, nonlinear generators, and stochastic representations of processes and semigroups are developed in \cite{criens2025stochastic,denk2020semigroup,neufeld_nutz_17}.
Related analytic results address generator characterizations, semigroup approximation, and comparison principles for viscosity solutions \cite{blessing2025gamma,dellacorte2026comparison,kraaij2022convergence}.

%The relationship between recursive valuation and model uncertainty is studied through multiple priors and dynamic risk measures in, for example, \cite{epstein2003recursive,cheridito2006dynamic}. Representations incorporating both model and discounting ambiguity are developed in \cite{follmerpenner2011,acciaiofollmerpenner2012,pennerreveilleac2015}. Thus, joint uncertainty about discounting and dynamics already has a substantial financial literature. Our question is how the local dynamics of an axiomatic valuation rule determine a concrete family of discounted weak models representing the entire rule. Related stochastic representations of semigroups and expectations on path space are established in \cite{criens2025representation,criens2025stochastic,nutz2013constructing}.

%We use the construction of canonical killed-path laws from a local specification in \cite{park2026global}. The probabilistic step proved here recovers the cumulative discounting and the underlying continuous state-process law from each such killed law. It identifies the pairs in \eqref{eq:pairfamily} with the laws used in that construction, and comparison then gives \eqref{eq:representation}. The recovery of discounting is therefore what turns the path-law representation into a description of the joint uncertainty over $(A,\mathbb Q)$ encoded by the valuation.

\section{Valuation and uncertainty}\label{sec:valuation}
The present paper studies two economic objects: dynamic sublinear valuation rules and robust valuations under uncertainty structures.
Throughout, $D\subset\mathbb R^d$ is a convex open domain admitting a smooth bounded convex exhaustion $(D_n)_{n\ge1}$ satisfying $\overline D_n\subset D_{n+1}\Subset D$.
Payoff functions are taken from $C_b(D)$, the space of bounded continuous functions on $D$.
This space is equipped with the mixed topology: a sequence converges if and only if it is uniformly bounded in the supremum norm and converges uniformly on compact subsets.
The space $C_b^\infty(D)$ consists of smooth bounded functions whose derivatives of every order are bounded.

\begin{definition}\label{def:valuation}
A dynamic sublinear valuation rule is a family of maps $\T_t:C_b(D)\to C_b(D)$, $t\ge0$, with $\T_0=\mathrm{id}$ and the following properties:
\begin{enumerate}[label=(V\arabic*),ref=(V\arabic*)]
\item\label{ax:sublinear} $\T_t$ is monotone and sublinear for every $t\ge0$.
\item\label{ax:contraction} $\|\T_t f\|_\infty\le\|f\|_\infty$ for every $t\ge0$ and $f\in C_b(D)$.
\item\label{ax:order} If $f_n\downarrow0$ in $C_b(D)$, then $\T_t f_n\downarrow0$ pointwise for every $t\ge0$.
\item\label{ax:continuity} $\T_{t_n}f\to\T_tf$ in the mixed topology whenever $t_n\to t\ge0$ and $f\in C_b(D)$.
\item\label{ax:semigroup} $\T_{s+t}=\T_s\T_t$ for all $s,t\ge0$.
\end{enumerate}
\end{definition}
Conditions~\ref{ax:order}--\ref{ax:continuity} ensure continuity with respect to claims and time, while \ref{ax:semigroup} expresses recursive valuation.
The norm bound in \ref{ax:contraction} is consistent with nonnegative discount rates.
Neither $\T_t1=1$ nor $\T_t(-f)=-\T_tf$ is required.
Condition~\ref{ax:sublinear} imposes monotonicity and sublinearity, which imply $-\T_t(-f)\le\T_tf$.
Thus, $\T$ and $-\T(-\cdot)$ provide upper and lower valuations, respectively.

We now introduce the concepts of uncertainty structures and their associated robust valuation rules.
Following \cite[Chapter~1]{pinsky1995positive}, we use the one-point compactification $\widehat D:=D\cup\{\Delta\}$ of $D$, equipped with a compatible metric.
The canonical path space $\widehat\Omega$ consists of continuous $\widehat D$-valued paths absorbed at $\Delta$, with Borel $\sigma$-field $\widehat{\mathcal F}$ and raw canonical filtration $(\widehat{\mathcal F}_t)_{t\ge0}$.
For the canonical process $X$, the exit and explosion times are
\begin{equation}\label{eq:exit_times}
\tau_n=\inf\{t\ge0:X_t\notin D_n\},\qquad \te=\lim_n\tau_n.
\end{equation}
The cemetery state $\Delta$ represents explosion. Once the process reaches $\Delta$, it remains there permanently and no further evolution takes place.
Functions on $D$ are extended by zero at $\Delta$.

The framework allows for joint uncertainty in discounting and the law of the underlying state process.
Let \(\mathfrak U\) consist of all pairs \((A,\mathbb Q)\), where \(A=(A_t)_{t\ge0}\) is an adapted, \([0,\infty]\)-valued, continuous, nondecreasing process on \(\hat\Omega\) with \(A_0=0\) and \(\mathbb Q\) is a probability measure on \(\hat\Omega\) such that
\begin{equation}
      \label{eq:prop:pairDUS_to_DSVR_1}
A_t<\infty
\textnormal{ for every } t\in [0,\tau_{\mathrm{exp}}) \,,
\quad
A_{\tau_{\mathrm{exp}}}\!\!=\infty
\textnormal{ on } \{\tau_{\mathrm{exp}}<\infty\},
\quad
\mathbb Q\text{-almost surely}.
\end{equation}
Two pairs \((A,\mathbb Q)\) and \((A',\mathbb Q')\) are identified if
\(\mathbb Q=\mathbb Q'\) and \(A,A'\) are indistinguishable under
\(\mathbb Q\).
We write \((A,\mathbb Q)\) for the corresponding equivalence class.
The space $[0,\infty]$ is equipped with the one-point compactification topology of $[0,\infty)$, with the convention $e^{-\infty}=0$.

\begin{definition}\label{def:rule}
For each $x\in D$, let $\mathcal U_x\subseteq\mathfrak U$
be a nonempty class of pairs $(A,\mathbb Q)$ satisfying
$\mathbb Q(A_0=0,\,X_0=x)=1$.
The family $\mathcal U=\{\mathcal U_x\}_{x\in D}$
is called an \emph{uncertainty structure}.
Its associated \emph{robust valuation rule} is the family
$\mathcal T^{\mathcal U}=\{\mathcal T_t^{\mathcal U}\}_{t\ge0}$
given by \eqref{eq:robust_valuation}.
\end{definition}

\section{Recovering Uncertainty from Valuation}\label{sec:recovery}
In this section, we show that a dynamic sublinear valuation rule admits a representation as a robust valuation under uncertainty.

\begin{definition}\label{def:representing}
Let $\mathcal T=\{\mathcal T_t\}_{t\ge0}$ be a dynamic sublinear valuation rule. We say that an uncertainty structure $\mathcal U=\{\mathcal U_x\}_{x\in D}$ represents $\mathcal T$ if $\mathcal T=\mathcal T^{\mathcal U}$.
\end{definition}

The recovery of uncertainty from valuation proceeds in four steps. 
Figure~\ref{fig:linear_valuation_uncertainty} illustrates the procedure. 
First, we extract the infinitesimal generator associated with the valuation rule $\{\mathcal T_t\}_{t\ge0}$ and describe its local behavior at each state $x\in D$ through a generating function $G$. 
Second, motivated by convex duality theory, we construct the support sets $\{A(x)\}_{x\in D}$ corresponding to the sublinear function $G(x,\cdot)$. 
Third, we construct an uncertainty structure $\mathcal U(G)$ from these support sets. 
Finally, we show that the robust valuation associated with the uncertainty structure $\mathcal U(G)$ coincides with the original valuation rule. 
The following subsections implement these steps in detail.

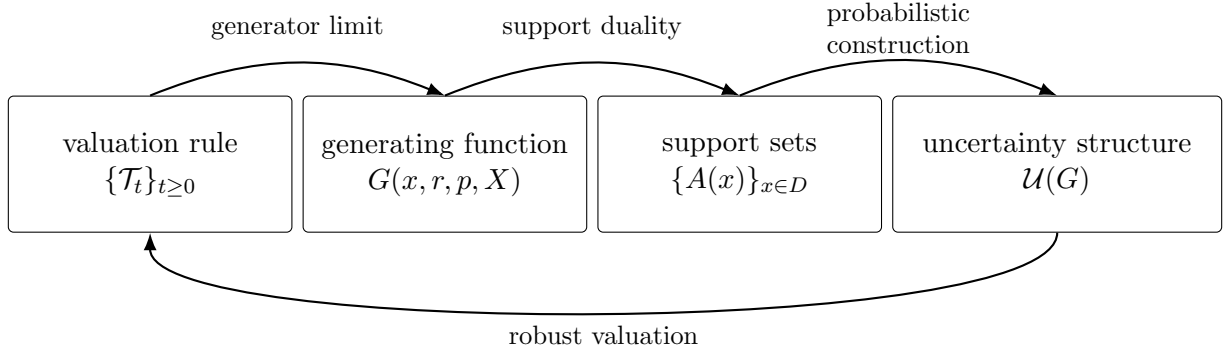
\begin{figure}[!ht]
\centering
\begin{tikzpicture}[
    scale=0.85,
    transform shape,
    font=\small,
    >=Latex,
    box/.style={
        draw,
        rounded corners=2pt,
        align=center,
        inner xsep=5pt,
        inner ysep=6pt,
        minimum height=18mm
    },
    mainbox/.style={
        box,
        text width=34mm
    },
    widebox/.style={
        box,
        text width=48mm
    },
    arr/.style={
        ->,
        thick
    },
    lab/.style={
        font=\footnotesize,
        fill=white,
        inner sep=1.5pt,
        align=center
    }
]

\node[mainbox] (T) at (0,0)
{
valuation rule\\
$\{\mathcal T_t\}_{t\ge0}$
};

\node[mainbox] (G) at (3.9,0)
{
generating function\\
$G(x,r,p,H)$
};

\node[mainbox] (A) at (7.8,0)
{
support sets\\
$\{A(x)\}_{x\in D}$
};

\node[box, text width=40mm] (P) at (12.0,0)
{
uncertainty structure\\
$\mathcal U(G)$
};

\draw[arr] (T.north) to[out=22,in=158] (G.north);
\draw[arr] (G.north) to[out=22,in=158] (A.north);
\draw[arr] (A.north) to[out=22,in=158] (P.north);

\draw[arr] (P.south) .. controls +(0,-1.4) and +(0,-1.4) .. (T.south);

\node[lab] at ($(T.north)!0.5!(G.north)+(0,0.9)$) {generator limit};
\node[lab] at ($(G.north)!0.5!(A.north)+(0,0.9)$) {support duality};
\node[lab] at ($(A.north)!0.5!(P.north)+(0,0.9)$) {probabilistic\\construction};

\node[lab] at ($(T.south)!0.5!(P.south)+(0,-1.35)$) {robust valuation};

\end{tikzpicture}
\caption{
Recovering Uncertainty from Valuation
}
\label{fig:linear_valuation_uncertainty}
\end{figure}

\subsection{From Valuation to Generating Function}
\label{subsec:generator}
 
The first step in recovering uncertainty from valuation is to extract the infinitesimal generator, motivated by classical semigroup theory. 
This infinitesimal generator characterizes the local behavior of the valuation rule.

\begin{definition}
    Let $\{\mathcal{T}_t\}_{t\geq0}$ be a dynamic sublinear valuation rule. The infinitesimal generator $\mathcal{G}:\mathcal{D}(\mathcal{G})\to C(D)$ is defined by
    \begin{align}\label{eq:valuation_generator}
        \mathcal{G}[f]:=\lim_{t\downarrow 0}\frac{\mathcal{T}_tf-f}{t}
    \end{align}
    where the domain $\mathcal{D}(\mathcal{G})$ consists of all functions $f\in C_b(D)$ for which the above limit exists with respect to the topology of local uniform convergence on $C(D)$.  
\end{definition}

We restrict our attention to valuation rules whose infinitesimal dynamics are local. 
Economically, the following assumption means that prices are driven by local market information: the instantaneous change at state $x$ depends only on nearby variations in fundamentals and payoffs. 
Thus, the generator $\mathcal G$ is restricted to the continuous-path, diffusion-type regime and excludes genuinely nonlocal effects such as jumps, crashes, or discrete policy interventions. 
This is a limitation of the present analysis, not of the valuation-based framework. 
Treating nonlocal generators would require a corresponding inverse theory for jump-type dynamics and is left for future work.

\begin{assumption}\label{ass:generator}
    Assume that $C_b^\infty(D)\subset \mathcal{D}(\mathcal{G})$ and the generator $\mathcal{G}$ is a local operator on $C_b^\infty(D)$, that is, if $f_1,f_2\in C_b^\infty(D)$ coincide in a neighborhood of
    $x\in D$, then $\mathcal{G}[f_1](x)=\mathcal{G}[f_2](x)$.
\end{assumption}

Under Assumption~\ref{ass:generator}, the infinitesimal generator admits a local pointwise representation: for each \(x\in D\), the value \(\mathcal G[f](x)\) depends only on \(x\), \(f(x)\), \(\nabla f(x)\), and \(\nabla^2 f(x)\). 
We denote by $\mathbb S^d$ the space of symmetric $d\times d$ matrices and by $\mathbb S^d_+$ its cone of positive semidefinite matrices.

\begin{proposition}\label{prop:generator}
    Let $\{\mathcal{T}_t\}_{t\geq0}$ be a dynamic sublinear valuation rule satisfying Assumption \ref{ass:generator}. Then there exists a function $G:D\times\mathbb{R}\times\mathbb{R}^d\times\mathbb S^d\to\mathbb{R}$ such that
    \begin{align}\label{eq:jet_generator}
        \mathcal{G}[f](x)=G(x,f(x),\nabla f(x),\nabla^2f(x))
        \end{align}
for all $f\in C_b^\infty(D).$
    Moreover, the function $G$ satisfies the following properties.
\begin{enumerate}[label=(G\arabic*), ref=(G\arabic*)]
    \item \label{item:G1}
    The function $G=G(x,r,p,H)$ is continuous in $(x, r, p, H)$ and sublinear in $(r, p, H)$.
    
    \item \label{item:G2}
    For all $(x,r,p)\in D\times\mathbb{R}\times\mathbb{R}^d$ and $H,Y\in\mathbb S^d$ with $H\geq Y$, 
      $G(x,r,p,H)\geq G(x,r,p,Y)$.    
    \item \label{item:G3}
    For all $(x,p,H)\in D\times\mathbb{R}^d\times\mathbb S^d$ and $r,s\in\mathbb{R}$ with $r\geq s$,
    $G(x,r,p,H) \leq G(x,s,p,H)$.
\end{enumerate}
\end{proposition}
\begin{proof}
The argument in \cite[proof of Proposition~2.1]{park2026global} applies to the present $C_b(D)$ setting and proves all assertions except joint continuity of $G$.
By \eqref{eq:valuation_generator}, $\mathcal G[f]$ is continuous for every $f\in C_b^\infty(D)$.
Combining \eqref{eq:jet_generator} with
\cite[Eq.~(2.3)]{park2026global}, and choosing $f$ with arbitrary
prescribed value, gradient, and Hessian at $x$, therefore proves
joint continuity of $G$.
\end{proof}

This function $G$ plays a central role throughout the paper.
It is a spatially local object determined by the valuation rule in a neighborhood of each point $x$, while the valuation rule itself is a global object determined by its behavior on the entire domain $D$.

\begin{definition}
The function $G:D\times\mathbb{R}\times\mathbb{R}^d\times\mathbb S^d\to\mathbb{R}$  in Proposition \ref{prop:generator} is called the generating function of the dynamic sublinear valuation rule $\{\mathcal{T}_t\}_{t\geq0}$.
\end{definition}

The parabolic comparison principle ensures that the local generating function $G$ uniquely determines the global valuation rule, a key step in recovering uncertainty from valuation.
Sufficient conditions for this principle are discussed in \cite{crandall1992user} and \cite[Appendix~C]{peng2019nonlinear}.
For precise definitions of viscosity supersolutions and subsolutions, see \cite{crandall1992user,crandall2000lp}.

\begin{comment} 
{\color{blue} Question: Is this enough (compared to comparison): 

\begin{assumption} %\label{ass:comparison}
Let \(\mathcal{D} \subset C_b (D)\) be such that for every \(f \in C_b (D)\) there exist sequences \((f_n)_{n = 1}^\infty \subset \mathcal{D}\) and \((a_n)_{n = 1}^\infty \subset \mathbb{R}\) such that \(f_n + a_n \downarrow f\) in \(C_b (D)\). 
Assume that the function
$G:D\times\mathbb R\times\mathbb R^d\times\mathbb S(d)\to\mathbb R$ 
satisfies the parabolic uniqueness principle for all initial functions from \(\mathcal{D}\), i.e., for every $T>0$ and \(f \in \mathcal{D}\), the PDE 
\begin{align} %\label{eq:pde}
\partial_t v = G(x,v,\nabla v,\nabla^2 v), \quad v(0,x)=f(x),
\end{align}
has at most one bounded (continuous) viscosity solution.
\end{assumption}

Remark: the set \(\mathcal{D}\) determines a valuation rule by its continuity from above property.
}

\end{comment}
 
\begin{assumption}\label{ass:comparison}
Assume that the function
$G:D\times\mathbb R\times\mathbb R^d\times\mathbb S(d)\to\mathbb R$
satisfies the parabolic comparison principle: for every $T>0$, if $v^+$ is a bounded viscosity supersolution and $v^-$ is a bounded viscosity subsolution of 
\begin{align}\label{eq:pde}
\partial_t v = G(x,v,\nabla v,\nabla^2 v), \quad v(0,x)=f(x),
\end{align}
on $[0,T)\times D$, then
$v^+\ge v^-$ on $[0,T)\times D$. 
\end{assumption}

The precise relationship between the valuation rule and the associated nonlinear PDE is given in the proposition below. 
We emphasize  that this result is fully model-free: it relies solely on the economic axioms imposed on the valuation rule and does not require any probabilistic assumptions. When a particular model is specified, the PDE \eqref{eq:pde} specializes to a Feynman--Kac-type equation. In particular, under the Black--Scholes specification, \eqref{eq:pde} reduces precisely to the classical Black--Scholes pricing PDE.

\begin{proposition}\label{prop:valuation_pde}   
Let   $\{\mathcal{T}_t\}_{t\geq0}$ be  a dynamic sublinear valuation rule on $C_b(D)$ satisfying Assumption \ref{ass:generator}, and let $G$ be its generating function.
Then, for any $f\in C_b(D)$, a function $v:[0,\infty)\times D\to\mathbb{R}$ defined by 
\begin{equation}
\label{eqn:PDE}
v(t,x):=\mathcal{T}_tf(x)
\end{equation}  
is a bounded viscosity solution to the PDE \eqref{eq:pde}.
If we further assume that $G$ satisfies Assumption~\ref{ass:comparison}, then $\{\mathcal T_t\}_{t\ge0}$ is the unique dynamic sublinear valuation rule satisfying Assumption \ref{ass:generator} with generating function $G$.
\end{proposition}
\begin{proof}
Conditions~\ref{ax:contraction} and~\ref{ax:continuity} imply
$v\in C_b([0,\infty)\times D)$ with $v(0,\cdot)=f$.
By Assumption~\ref{ass:generator},
\eqref{eq:valuation_generator}, and~\eqref{eq:jet_generator},
the inequalities in \cite[Eqs.~(3.48) and~(5.6)]{park2026global}
hold with equality.
The proofs of \cite[Lemmas~3.18 and~5.4]{park2026global}
therefore apply in the present $C_b(D)$ setting and show that
$v$ is a viscosity solution of \eqref{eq:pde}.
Under Assumption~\ref{ass:comparison}, comparison identifies
the solutions corresponding to any two valuation rules with
generating function $G$ and the same initial payoff,
proving uniqueness.
\end{proof}

The generating function $G$ yields an analytic description of the valuation rule through \eqref{eq:pde}. While recovering uncertainty from valuation, \eqref{eqn:PDE}--\eqref{eq:pde} make explicit how $G$ determines the evolution of the valuation function under Assumption~\ref{ass:comparison}.
In other words, it provides the analytic bridge from the infinitesimal object recovered from the valuation rule to the full dynamic valuation itself.
Because smooth solutions need not exist in degenerate cases, and because within our axiomatic framework it is not known a priori whether the recovered generating function $G$ is degenerate or nondegenerate, we work entirely within the viscosity-solution framework. 

\subsection{From Generating Function to Support Sets}\label{subsec:support}  
We next introduce the support sets associated with a generating function $G$. 
For $V=(C,B,\Sigma)\in \mathbb R\times\mathbb R^d\times\mathbb S^d$, let 
\[
L^V(r,p,H)
:=
\frac12\operatorname{tr}(\Sigma H)
+B\cdot p
+Cr,
\qquad
(r,p,H)\in\mathbb R\times\mathbb R^d\times\mathbb S^d.
\]
For each $x\in D$, the support set of $G(x,\cdot)$ is defined as
\begin{align}\label{eq:support} 
A(x)
:=
\Big\{
V\in(-\infty,0]\times\mathbb R^d\times\mathbb S^d_+
:
G(x,W)\ge L^V(W)
\text{ for all }
W\in\mathbb R\times\mathbb R^d\times\mathbb S^d
\Big\}.
\end{align}
Since the map $U\mapsto G(x,U)$ is sublinear, the classical dual representation theorem for sublinear functionals (see, e.g., \cite{rockafellar2015convex}) implies that $A(x)$ is nonempty, compact, and convex. Moreover, $G$ admits the representation
\begin{equation}\label{eq:support_representation}
G(x,U)=\max_{V\in A(x)}L^V(U).
\end{equation}
Continuity in \ref{item:G1} and \eqref{eq:support_representation} imply that $A$ has a closed graph and is locally bounded. %evaluating the right-hand side at the positive and negative basis vectors bounds each coordinate uniformly on compact subsets of $D$.

\subsection{From Support Sets to Uncertainty}
\label{subsec:pair_construction}
 
 We now pass from the support sets to a probabilistic uncertainty structure.
A progressively measurable process
\[
\beta=(C,B,\Sigma):[0,\infty)\times\widehat\Omega
\to
(-\infty,0]\times\mathbb{R}^d\times\mathbb S^d_+
\]
is called a \emph{coefficient field}.
A coefficient field $\beta$ is \emph{admissible} if
\begin{equation}\label{eq:admissibility}
    \begin{aligned}
         &\beta(t,\omega)\in A(\omega(t))
        \text{ for } t<\tau_{\mathrm{exp}}(\omega) \;\text{ and  }\; \beta(t,\omega)=0
         \text{ for } t\ge\tau_{\mathrm{exp}}(\omega).
    \end{aligned}
\end{equation}
The first condition in \eqref{eq:admissibility} is equivalent to
\begin{align}\label{eq:coefficientadmissible}
    L^\beta(t,\omega,U)
    \le
    G(\omega(t),U)
      \text{ for all }(t,\omega,U)\text{ with }t<\tau_{\mathrm{exp}}(\omega),
\end{align}
where $L^\beta(t,\omega,U):=L^{\beta(t,\omega)}(U)$ and $U=(r,p,H)$.
Thus, admissibility means that the linear functional associated with $\beta$ is pointwise dominated by the generating function $G$.
We write 
$\mathcal B_{\mathrm{ad}}(G)$ for the collection of all admissible
coefficient fields. 
For a coefficient field $\beta=(C,B,\Sigma)$, the value
 $\beta(t,\omega)\in (-\infty,0]\times\mathbb{R}^d\times\mathbb S^d_+$ 
at the time-path pair $(t,\omega)$ is called the local characteristic of $\beta$ at $(t,\omega)$.
The support set $A(x)$ therefore represents the collection of all possible local characteristics of admissible coefficient fields at time-path pairs satisfying $\omega(t)=x$.

For any $x\in D$ and any admissible coefficient field $\beta=(C,B,\Sigma)$,
we construct a cumulative discounting process and a class of laws for the underlying state process. 
The cumulative discounting process is determined by the $C$-component of $\beta$.
Let $k:=-C$ and define
\begin{equation}\label{eq:clock}
A_t^k:=\int_0^t k(s,X)\mathbb I_{\{s<\te\}}\,ds,
\qquad t\ge0.
\end{equation}
The integral in \eqref{eq:clock} is defined pathwise, independently of a probability measure. %At and after finite explosion, it is the increasing limit of the pre-explosion integrals.

The coefficients $\gamma:=(B,\Sigma)$ determine a class of possible laws for the underlying state process.
We denote by $\mathcal P_x(L^\gamma)$ the collection of probability measures $\mathbb Q$ on $(\widehat\Omega,\widehat{\mathcal F})$ satisfying $\mathbb Q(X_0=x)=1$ such that, for every $f\in C_c^\infty(D)$ and $n\ge1$,
\begin{align}\label{eq:state_mp}
M_t^n:=f(X_{t\wedge\tau_n})-\int_0^{t\wedge\tau_n}
L^\gamma(u,X,\nabla f(X_u),\nabla^2f(X_u))\,du
\end{align}
is a continuous $\mathbb Q$-martingale with respect to $(\widehat{\mathcal F}_t)_{t\ge0}$, where
\begin{equation}\label{eqn:L}
L^\gamma(t,\omega,p,H)
:=\frac12\operatorname{tr}(\Sigma(t,\omega)H)+B(t,\omega)\cdot p.
\end{equation}
This generalized martingale problem may have no solution or multiple solutions.  Refer to \cite{pinsky1995positive} for further details.
Each element \(\mathbb Q\in\mathcal P_x(L^\gamma)\) represents a possible law of the underlying state process. 

%This construction leads to the definition of uncertainty structures. The family $\mathcal U_x(G)$ introduced below consists of pairs of a cumulative discounting process and a law  for the underlying state process associated with admissible coefficient fields.

\begin{definition}\label{def:uncertainty_G}
For each $x\in D$, define
\begin{equation}\label{eq:pairfamily}
\mathcal U_x(G):=
\left\{(A^k,\mathbb Q):\;(-k,\gamma)\in\mathcal B_{\mathrm{ad}}(G),\quad
\mathbb Q\in\mathcal P_x(L^\gamma)\right\}.
\end{equation}
The family $\mathcal U(G):=\{\mathcal U_x(G)\}_{x\in D}$ is called the \emph{uncertainty structure associated with $G$}.
\end{definition}
Lemma~\ref{lem:uncertainty_structure} shows that $\mathcal U(G)$ is indeed an uncertainty structure under Assumption~\ref{ass:lyapunov}.

\subsection{Completing the converse direction}\label{subsec:stochasticrepn1}

We have constructed the uncertainty structure $\mathcal U(G)$ from the dynamic valuation rule $\{\mathcal T_t\}_{t\ge0}$.
It remains to show that the associated robust valuation rule $\{\mathcal T_t^{\mathcal U(G)}\}_{t\ge0}$ coincides with the original rule.
This equality completes the cycle
\begin{align}\label{eq:mainchain}
	\{\mathcal T_t\}_{t\ge0} \;\to\; G \;\to\; A(\,\cdot\,) \;\to\; \mathcal U(G) \;\to\; \{\mathcal T_t^{\mathcal U(G)}\}_{t\ge0}= \{\mathcal T_t\}_{t\ge0}\,,
\end{align}
illustrated in Figure~\ref{fig:linear_valuation_uncertainty}.

\begin{assumption}\label{ass:lyapunov}
	Assume that the function $G:D\times\mathbb{R}\times\mathbb{R}^d\times\mathbb S^d\to\mathbb{R}$ satisfies a Lyapunov-type condition:	
	There exist a positive function $\phi\in C^2(D)$ and a constant $C$ such that
	$\phi(x)\to \infty$ as $x\to\partial D$\footnote{For $f:D\to\mathbb R$, we write $f(x)\to\infty$ as $x\to\partial D$ if $\lim_{n\to\infty}\inf_{x\in D\setminus D_n} f(x)=\infty$.} and for all $x\in D$, 
	\begin{equation}\label{eq:lyapunov}
	G(x,\phi(x),\nabla \phi(x),\nabla^2 \phi(x))\le C \phi(x).
	\end{equation} 
\end{assumption}

\begin{lemma}\label{lem:uncertainty_structure}
Suppose that $G$ satisfies conditions~\ref{item:G1}--\ref{item:G3}
and Assumption~\ref{ass:lyapunov}.
Then the family $\mathcal U(G)$ in \eqref{eq:pairfamily}
is an uncertainty structure as defined in
Definition~\ref{def:rule}.
%{\color{red} More precisely, $\mathcal U_x(G)\neq\varnothing$ for every $x\in D$, and each $(A,\mathbb Q)\in\mathcal U_x(G)$ consists of a probability measure $\mathbb Q$ on $(\widehat\Omega,\widehat{\mathcal F})$ satisfying $\mathbb Q(X_0=x)=1$ and an $(\widehat{\mathcal F}_t)_{t\ge0}$-progressively measurable, $[0,\infty]$-valued process $A$ with $A_0=0$ and nondecreasing paths. The paths of $A$ are $\mathbb Q$-almost surely continuous in $[0,\infty]$, and the pair $(A,\mathbb Q)$ satisfies \eqref{eq:prop:pairDUS_to_DSVR_1}.}
\end{lemma}

\begin{proof}
Fix $x\in D$.
We first verify \eqref{eq:prop:pairDUS_to_DSVR_1} for every pair $(A^k,\mathbb Q)\in\mathcal U_x(G)$. We may replace $\phi$ and $C$ by $\phi+1$ and $C\vee0$: property~\ref{item:G3} preserves \eqref{eq:lyapunov}. We may therefore assume $\phi\ge1$ and $C\ge0$. For any admissible $\beta=(-k,\gamma)$ and $\mathbb Q\in\mathcal P_x(L^\gamma)$, equations~\eqref{eq:state_mp}, \eqref{eq:coefficientadmissible}, and~\eqref{eq:lyapunov} imply that $e^{-C(t\wedge\tau_n)-A^k_{t\wedge\tau_n}}\phi(X_{t\wedge\tau_n})$ is a nonnegative $\mathbb Q$-supermartingale. Here the use of $\phi\in C^2(D)$ in \eqref{eq:state_mp} is justified by a cutoff and smooth approximation on $\overline D_n$. With $m_n:=\inf_{D\setminus D_n}\phi$, optional sampling gives, for $n$ sufficiently large that $x\in D_n$,
\begin{equation}\label{eq:extraction_clock_bound}
 \mathbb Q(\tau_n\le T,\ A^k_{\tau_n}\le R)
 \le \frac{e^{CT+R}\phi(x)}{m_n},\qquad 0\le T,R<\infty.
\end{equation}
Since $m_n\to\infty$, letting $n\to\infty$ in \eqref{eq:extraction_clock_bound} and then taking the union over integer $T,R$ shows that $A^k_{\tau_{\mathrm{exp}}}=\infty$ on $\{\tau_{\mathrm{exp}}<\infty\}$, $\mathbb Q$-almost surely. Local boundedness of $k$ gives $A_t^k<\infty$ for $t<\tau_{\mathrm{exp}}$, proving \eqref{eq:prop:pairDUS_to_DSVR_1}.

By \eqref{eq:clock}, $A^k$ is progressive, starts at zero, and has nondecreasing paths. Local boundedness of $k$ gives continuity before explosion, while the divergence just proved gives continuity at and after finite explosion in $[0,\infty]$. Moreover, $\mathbb Q(X_0=x)=1$ by the definition of $\mathcal P_x(L^\gamma)$. Finally, Proposition~\ref{prop:extraction}, proved in Section~\ref{subsec:extraction}, establishes that $\mathcal U_x(G)$ is nonempty for every $x\in D$.
\end{proof}

The next theorem provides a stochastic representation of dynamic sublinear valuation rules and constitutes one of the main results of this paper.
The proof is given in Section~\ref{subsec:extraction}.
Recall that  
\begin{align}\label{eq:representation}
\begin{split}
\mathcal{T}_t^{\mathcal U(G)} f(x)
:=&
\sup_{(A,\mathbb{Q})\in\mathcal{U}_x(G)}
\mathbb{E}^{\mathbb{Q}}\!\left[
e^{-A_t}f(X_t)\mathbb{I}_{\{\tau_{\mathrm{exp}}>t\}}
\right] \\
=&
\sup_{(-k,\gamma)\in\mathcal{B}_{\mathrm{ad}}(G)}
\sup_{\mathbb{Q}\in\mathcal{P}_x(L^\gamma)}
\mathbb{E}^{\mathbb{Q}}\!\left[
e^{-\int_0^t k_s\,ds}
f(X_t)\mathbb{I}_{\{\tau_{\mathrm{exp}}>t\}}
\right]\,,
\end{split}
\end{align}
for $(t,x)\in[0,\infty)\times D$ and $f\in C_b(D)$.

\begin{theorem}\label{thm:representation}
	Let $\mathcal T:=\{\mathcal T_t\}_{t\ge0}$ be a dynamic sublinear valuation rule  satisfying Assumption~\ref{ass:generator}, and let $G$ be its generating function. 
	Suppose that $G$ satisfies Assumptions~\ref{ass:comparison} and~\ref{ass:lyapunov}. 
	Then \(\mathcal T\) coincides with the robust valuation rule under the uncertainty structure \(\mathcal U(G)\), that is,
\begin{equation}\label{eq:valuation_identity}
\mathcal T=\mathcal T^{\mathcal U(G)}.
\end{equation}
\end{theorem}

This completes the cycle in Figure \ref{fig:linear_valuation_uncertainty}.
The procedure recovers the latent models $(A,\mathbb Q)$ by specifying cumulative discounting and the law of the state process.
Consequently. the valuation rule thus reveals the uncertainty structure under which payoffs are evaluated.

A key insight of our uncertainty identification theory is that a global uncertainty structure can be recovered from local information encoded in a valuation rule.
The generating function and its support set are local objects: their values at a state \(x\in D\) are determined by information in a neighborhood of \(x\).
By contrast, uncertainty structures and robust valuation rules are global objects, since their values depend on the evolution of the state process over the entire domain \(D\).
Theorem~\ref{thm:representation} shows that 
piecing together these local objects extracted from a dynamic valuation rule recovers the latent uncertainty structure \(\mathcal U(G)\).

\begin{remark} 
\label{rem:option_identification}
Suppose that $D=(0,\infty)$ and that $\mathcal T$ is a dynamic
sublinear valuation rule satisfying
Assumption~\ref{ass:generator}, with generating function $G$.
Write $p_K(y):=(K-y)_+$ for the put payoff with strike \(K>0\).

\emph{General identification.} Suppose that, for some $\varepsilon>0$, the values
$\mathcal T_t f(x)$ are known for every $x\in D$,
$0<t<\varepsilon$, and every finite signed put portfolio
$f=\sum_{j=1}^n\beta_jp_{K_j}$.
These observations determine $G$ without specifying its
form or the discounting mechanism in advance.
Indeed, for $\psi\in C_c^\infty(D)$, integration by parts gives
\[
\psi(y)=\int_0^\infty \psi''(K)p_K(y)\,dK.
\]
Since $\psi''$ has compact support and
$\|p_K-p_L\|_\infty\le |K-L|$, Riemann sums approximate
$\psi$ uniformly by signed put portfolios.
Sublinearity and \ref{ax:contraction} imply
$\|\mathcal T_t f-\mathcal T_t g\|_\infty
\le\|f-g\|_\infty$,
so these observations determine $\mathcal T_t\psi(x)$.
For any $(r,p,H)\in\mathbb R^3$, choose such a $\psi$
with $(\psi(x),\psi'(x),\psi''(x))=(r,p,H)$.
Then \eqref{eq:jet_generator} yields
\[
G(x,r,p,H)
=
\lim_{t\downarrow0}
\frac{\mathcal T_t\psi(x)-r}{t}.
\]
The portfolio approximation is taken at fixed $t$ before
letting $t\downarrow0$.
Thus the observations determine $G$ and its associated support sets
$A(x)$.
Under Assumptions~\ref{ass:comparison}
and~\ref{ass:lyapunov},
Theorem~\ref{thm:representation} identifies
$\mathcal T=\mathcal T^{\mathcal U(G)}$.

\emph{Identification under structural restrictions.}
Structural information about the generator can reduce
the required family of payoffs.
For example, suppose that  
\[
G(x,r,p,H)
=\frac{1}{2} \Big( x^2 \overline{\sigma}^2 (x) H_+ + x^2 \underline{\sigma}^2 (x) H_- \Big) = \max_{\underline{\sigma} (x) \leq s \leq \overline{\sigma} (x)} \tfrac{1}{2} x^2 s^2 H, 
\]
where $0<\underline\sigma\le\overline\sigma$ and
$\underline\sigma(e^\cdot),\overline\sigma(e^\cdot)$
are bounded and Lipschitz continuous.
The comparison assumption follows from
\cite[Theorem~C.2.5]{peng2019nonlinear}
after a logarithmic change of variables,
and the Lyapunov assumption holds with
$\phi(x)=x^2+x^{-1}$.
Fix $x_0>0$ and set $
P^{\mathrm{ask}}(T,K)=\mathcal T_Tp_K(x_0)$ and $P^{\mathrm{bid}}(T,K)=-\mathcal T_T(-p_K)(x_0)$.
Convex-order comparison, as in the argument of
\cite[Lemma~6.1]{criensniemann2024markov},
shows that these are the put price surfaces of the local
volatility models with coefficients
$\overline\sigma$ and $\underline\sigma$, respectively.
For sufficiently smooth surfaces with positive strike curvature,
Dupire's formula \cite{dupire} gives 
$\overline\sigma(K)^2=2\partial_TP^{\mathrm{ask}}(T,K)/(K^2\partial_{KK}P^{\mathrm{ask}}(T,K))$ and $\underline\sigma(K)^2=2\partial_TP^{\mathrm{bid}}(T,K)/(K^2\partial_{KK}P^{\mathrm{bid}}(T,K))$.
%\[\overline\sigma(K)^2=\frac{2\partial_TP^{\mathrm{ask}}(T,K)}     {K^2\partial_{KK}P^{\mathrm{ask}}(T,K)},\qquad\underline\sigma(K)^2=\frac{2\partial_TP^{\mathrm{bid}}(T,K)}     {K^2\partial_{KK}P^{\mathrm{bid}}(T,K)}.\]
Compatible bid and ask put surfaces therefore determine
$G$ and $\mathcal T$ within this class.
Compatibility requires, in particular, that these ratios
be independent of $T$ and satisfy the prescribed ordering.

The first result uses valuations of entire portfolios,
which cannot generally be inferred from individual put
bid and ask prices.
The second uses structural restrictions on $G$ to reduce
these data requirements; agreement with  put prices alone does not
establish those restrictions.
Both results concern recovery from idealized valuation data.
Model selection and calibration to finitely many, possibly noisy,
market quotations are separate problems not addressed here.
\end{remark}

\subsection{Proof of Theorem~\ref{thm:representation}}\label{subsec:extraction}
Let $\widetilde\Omega$ be the Borel subspace of the Skorokhod space
$D([0,\infty),\widehat D)$ consisting of paths continuous in $D$
before a lifetime $\zeta$ and absorbed at $\Delta$ afterwards.
A terminal jump from a point of $D$ to $\Delta$ is allowed.
Write $Y$ for the coordinate process,
$(\widetilde{\mathcal F}_t)_{t\ge0}$ for its raw canonical filtration,
and $\rho_n=\inf\{t\ge0:Y_t\notin D_n\}$.
We define the continuous explosion time on $\widetilde\Omega$ by
\begin{equation}\label{eq:killed_explosion_time}
\tau_{\mathrm{exp}}
:=
\begin{cases}
\zeta,
& \zeta<\infty \text{ and } Y_{\zeta-}=\Delta,\\
\infty,
& \text{otherwise},
\end{cases}
\end{equation}
with the convention $Y_{0-}=Y_0$.
This agrees with $\tau_{\mathrm{exp}}$
on $\widehat\Omega$.
The exit times satisfy $\rho_n\uparrow\zeta$.
In particular, a path with a terminal jump from $D$ to $\Delta$
has $\zeta<\infty$ and $\tau_{\mathrm{exp}}=\infty$.
For a path with a killing jump, let $r(Y)$ follow $Y$ until that
jump and remain at its left limit afterwards; for a path without
such a jump put $r(Y)=Y$.
Before the lifetime, $r(Y)$ and $Y$ have the same history.
Accordingly, progressive coefficient fields can be transferred
between $\widehat\Omega$ and $\widetilde\Omega$ by restriction
and by evaluation at $r(Y)$ before $\zeta$, with value zero afterwards.

For $G$ satisfying conditions~\ref{item:G1}--\ref{item:G3}
and \eqref{eq:lyapunov}, and $x\in D$, let $\mathcal P_x(G)$
consist of the probability measures $\mathbb P$ on
$(\widetilde\Omega,\mathcal B(\widetilde\Omega))$ such that
$\mathbb P(Y_0=x)=1$ and
\begin{equation}\label{eq:killed_supermp}
f(Y_{t\wedge\rho_n})-f(x)
-\int_0^{t\wedge\rho_n}
G(Y_s,f(Y_s),\nabla f(Y_s),\nabla^2f(Y_s))\ind_{\{s<\zeta\}}\,ds
\end{equation}
is a $\mathbb P$-supermartingale with respect to
$(\widetilde{\mathcal F}_t)_{t\ge0}$ for every $n\ge1$
and $f\in C_b^\infty(D)$, as in
\cite[Definition~2.5]{park2026global}.
For $t\ge0$ and $f\in C_b(D)$, define
\begin{equation}\label{eq:canonical_value}
\mathcal V_t^Gf(x)
=\sup_{\PP\in\mathcal P_x(G)}
\E^{\PP}[f(Y_t)\ind_{\{t<\zeta\}}].
\end{equation}

For a pair $(A,\mathbb Q)\in\mathfrak U$, let $E$ be an independent unit exponential random variable on a product extension of
$(\widehat\Omega,\mathbb Q)$.   
Set $\zeta_A:=\inf\{t\ge0:A_t\ge E\}$, with $\inf\varnothing=\infty$, and define $Y_t^A:=X_t$ for $t<\zeta_A$ and $Y_t^A:=\Delta$ for $t\ge\zeta_A$.
We define $\Phi(A,\mathbb Q)$ as the law of $Y^A$ on
$\widetilde\Omega$ under this product extension.
Condition~\eqref{eq:prop:pairDUS_to_DSVR_1} ensures that
$\zeta_A<\tau_{\mathrm{exp}}$ on
$\{\tau_{\mathrm{exp}}<\infty\}$ almost surely.
Replacing $A$ by a $\mathbb Q$-indistinguishable process leaves $\Phi(A,\mathbb Q)$ unchanged, so $\Phi$ is well defined on the equivalence classes in $\mathfrak U$. 

The proof of Theorem~\ref{thm:representation} combines
the representation by killed laws in
\cite{park2026global} with the following
proposition, which recovers the cumulative discounting process
and the state-process law from each $\mathbb P\in\mathcal P_x(G)$.

\begin{proposition}\label{prop:extraction}
Suppose that $G$ satisfies \ref{item:G1}--\ref{item:G3} and Assumption~\ref{ass:lyapunov}. Then
\begin{equation}\label{eq:extraction_family}
 \Phi\bigl(\mathcal U_x(G)\bigr)=\mathcal P_x(G),\qquad x\in D.
\end{equation}
Each $\mathbb P\in\mathcal P_x(G)$ determines a unique pair
$(A,\mathbb Q)\in\mathcal U_x(G)$ satisfying
$\Phi(A,\mathbb Q)=\mathbb P$, with
cumulative discounting processes identified up to
$\mathbb Q$-indistinguishability.
In particular, $\mathcal U_x(G)$ is nonempty.
\end{proposition}

\begin{proof}
We first prove $\mathcal P_x(G)\subseteq\Phi\bigl(\mathcal U_x(G))$ and $\mathcal U_x(G)\neq\varnothing$.
Fix $\mathbb P\in\mathcal P_x(G)$.
By \cite[Proposition~3.2]{park2026global}, choose a progressive
field $\beta=(-k,B,\Sigma)$ satisfying \eqref{eq:admissibility}
on the killed space, with $\zeta$ in place of $\te$,
and set $\gamma=(B,\Sigma)$.
By \cite[Lemma~3.4]{park2026global},
$\mathbb P(\zeta<\infty,\ Y_{\zeta-}=\Delta)=0$.
Consequently, $\overline Y:=r(Y)\in C([0,\infty),D)$,
$\mathbb P$-almost surely.
With $N_t=\mathbb I_{\{\zeta\le t\}}$,
\cite[Proposition~3.2(i)]{park2026global} identifies its compensator as
\begin{equation}\label{eq:extraction_compensator}
 \Lambda_t=\int_0^t k(s,Y)\mathbb I_{\{s<\zeta\}}\,ds.
\end{equation}
Under the identification of coefficient fields before the lifetime,
$\beta(s,Y)=\beta(s,\overline Y)$ for $s<\zeta$.
By \cite[Proposition~3.2(ii)]{park2026global}, the drift and
covariance densities of $\overline Y_{\cdot\wedge\rho_n}$
before $\rho_n$ are $B(s,Y)\ind_{\{s<\zeta\}}$ and
$\Sigma(s,Y)\ind_{\{s<\zeta\}}$, respectively.

Put $\sigma_n=n\wedge\rho_n$, taking $n$ large enough that $x\in D_n$.
By \eqref{eq:extraction_compensator} and the product rule, we have
\begin{equation}\label{eq:extraction_density}
 Z_t:=e^{\Lambda_t}\mathbb I_{\{t<\zeta\}},\qquad
 Z_{t\wedge\sigma_n}
 =1-\int_0^{t\wedge\sigma_n}e^{\Lambda_s}\,d(N_s-\Lambda_s).
\end{equation}
Local boundedness of $k$ gives
$0\le Z_{t\wedge\sigma_n}\le e^{c_n n}$
for some $c_n<\infty$, so $Z^{\sigma_n}$ is a true martingale.
Consequently,
$\frac{d\mathbb Q^n}{d\mathbb P}
       |_{\widetilde{\mathcal F}_{\sigma_n}}=Z_{\sigma_n}$
defines a probability measure.
By construction, these measures satisfy
$\mathbb Q^n(\zeta\le\sigma_n)=0$ and
$\mathbb Q^{n+1}|_{\widetilde{\mathcal F}_{\sigma_n}}=\mathbb Q^n$.
Under $\mathbb Q^n$, for every $f\in C_c^\infty(D)$,
\begin{equation}\label{eq:extraction_stopped_martingale}
 f(\overline Y_{t\wedge\sigma_n})-f(x)
 -\int_0^{t\wedge\sigma_n}
       L^\gamma(s,\overline Y,\nabla f(\overline Y_s),\nabla^2f(\overline Y_s))\,ds
 \quad\text{is a $\mathbb Q^n$-martingale.}
\end{equation}
Indeed, the continuous local martingale part of
$\overline Y_{\cdot\wedge\rho_n}$ obtained from \cite[Proposition~3.2(ii)]{park2026global} has zero quadratic covariation with $Z^{\sigma_n}$, since \eqref{eq:extraction_density} shows that the latter has no continuous local martingale part.
Integration by parts and Bayes' formula therefore preserve these drift and covariance densities.
Since $\mathbb Q^n(\sigma_n<\zeta)=1$, It\^o's formula yields
\eqref{eq:extraction_stopped_martingale}.

On $C([0,\infty),D)$, define
$S_nw:=w(\cdot\wedge(n\wedge\tau_n(w)))$ and let
$\mu_n:=\mathbb Q^n\circ
(\overline Y_{\cdot\wedge\sigma_n})^{-1}$.
The preceding consistency gives
$\mu_{n+1}\circ S_n^{-1}=\mu_n$.
Disintegrating these laws over the maps $S_n$ on the standard
Borel path spaces and iterating the resulting kernels produces
random paths $W^n$ with laws $\mu_n$ satisfying
$S_nW^{n+1}=W^n$ almost surely.
Set $\eta_n:=n\wedge\tau_n(W^n)$,
$\xi:=\lim_n\eta_n$, and $X_t:=W_t^n$ whenever $t<\eta_n$.
These definitions are consistent and give a continuous process
on $[0,\xi)$.
On $\{\xi<\infty\}$, for all sufficiently large $n$,
$\eta_n=\tau_n(W^n)$ and
$W^n_{\eta_n}\in\partial D_n\subset D_{n+1}$;
hence $\eta_n<\eta_{n+1}$.

For $f\in C_c^\infty(D)$,
\eqref{eq:extraction_stopped_martingale} and the preceding construction show that, before $\eta_n$, the drift density of $f(X)$ and the quadratic variation density of its continuous local martingale part are
$L^\gamma(s,X,\nabla f(X_s),\nabla^2f(X_s))$ and
$\nabla f(X_s)^\top\Sigma(s,X)\nabla f(X_s)$, respectively.
Their absolute values are bounded by a constant $c_f<\infty$,
independently of $n$, because the derivatives of $f$ have compact
support.
Thus, up to $T\wedge\eta_n$, the drift variation and martingale
quadratic variation are bounded by $c_fT$.
The martingale convergence theorem gives a finite limit of
$f(X_t)$ as $t\uparrow\xi$ on $\{\xi\le T\}$.
Since $f(X_{\eta_n})=0$ for all sufficiently large $n$ on this event,
$\lim_{t\uparrow\xi}f(X_t)=0$.
Applying this to a countable collection
$f_m\in C_c^\infty(D)$ with $f_m=1$ on $\overline D_m$
shows that, on $\{\xi<\infty\}$, for every $m$ there exists
$s_m<\xi$ such that $X_t\notin\overline D_m$ whenever
$s_m<t<\xi$.
Hence $X_t\to\Delta$ in $\widehat D$ as $t\uparrow\xi$.
Defining $X_t=\Delta$ for $t\ge\xi$ gives a law $\mathbb Q$
on $\widehat\Omega$, with $\xi=\tau_{\mathrm{exp}}$.
By localization, \eqref{eq:extraction_stopped_martingale}
implies \eqref{eq:state_mp}, so
$\mathbb Q\in\mathcal P_x(L^\gamma)$.

By \eqref{eq:extraction_clock_bound} and the local boundedness of $k$,
the pair $(A^k,\mathbb Q)$ satisfies
\eqref{eq:prop:pairDUS_to_DSVR_1} and belongs to $\mathcal U_x(G)$
by \eqref{eq:pairfamily}.
For every bounded $\widehat{\mathcal F}_t$-measurable $H$,
a change of measure with density $Z_{\sigma_n}$ and conditioning at time $t$ yield
$\mathbb E^{\mathbb Q}
[e^{-A_t^k}H(X)\mathbb I_{\{t<n\wedge\tau_n\}}]
=\mathbb E^{\mathbb P}
[H(\overline Y)\mathbb I_{\{t<\sigma_n\}}]$.
Applying bounded convergence as $n\to\infty$ gives
\begin{equation}\label{eq:extraction_survival}
 \mathbb E^{\mathbb Q}
   [e^{-A_t^k}H(X)\mathbb I_{\{t<\tau_{\mathrm{exp}}\}}]
 =\mathbb E^{\mathbb P}
   [H(\overline Y)\mathbb I_{\{t<\zeta\}}].
\end{equation}
By the definition of $\Phi$, the left-hand side equals
$\mathbb E^{\Phi(A^k,\mathbb Q)}
[H(r(Y))\mathbb I_{\{t<\zeta\}}]$.

For $0=t_0<t_1<\cdots<t_m$ and $1\le j\le m$, use
$\mathbb I_{\{t_{j-1}<\zeta\le t_j\}}
=\mathbb I_{\{t_{j-1}<\zeta\}}
-\mathbb I_{\{t_j<\zeta\}}$
with $H=\prod_{i<j}g_i(X_{t_i})$, where the $g_i$ are bounded
Borel functions on $D$.
Together with \eqref{eq:extraction_survival} at $t_m$
for $H=\prod_{i=1}^m g_i(X_{t_i})$, these identities determine
all finite-dimensional distributions of the killed process.
Therefore $\Phi(A^k,\mathbb Q)=\mathbb P$.
This proves $\mathcal P_x(G)\subseteq\Phi\bigl(\mathcal U_x(G))$ and $\mathcal U_x(G)\neq\varnothing$.

Next we prove $\Phi\bigl(\mathcal U_x(G))\subseteq\mathcal P_x(G)$. 
Fix any admissible $\beta=(-k,B,\Sigma)$ and
$\mathbb Q\in\mathcal P_x(L^\gamma)$.
Equation~\eqref{eq:extraction_clock_bound} gives
$A^k_{\tau_{\mathrm{exp}}}=\infty$ on
$\{\tau_{\mathrm{exp}}<\infty\}$, $\mathbb Q$-almost surely.
Thus, on the product extension used to define $\Phi$,
$\zeta_{A^k}<\tau_{\mathrm{exp}}$ whenever
$\tau_{\mathrm{exp}}<\infty$, almost surely.
For $\mathbb P:=\Phi(A^k,\mathbb Q)$, the definition of $\Phi$,
\eqref{eq:state_mp}, and the product rule show that, for every
$n\ge1$ and $f\in C_b^\infty(D)$, the process
\[
f(Y_{t\wedge\rho_n})-f(x)
-\int_0^{t\wedge\rho_n}
\Bigl[
L^\gamma(s,Y,\nabla f(Y_s),\nabla^2f(Y_s))
-k(s,Y)f(Y_s)
\Bigr]\mathbb I_{\{s<\zeta\}}\,ds,
\qquad t\ge0,
\]
is a $\mathbb P$-martingale with respect to
$(\widetilde{\mathcal F}_t)_{t\ge0}$.
By \cite[Corollary~3.3]{park2026global},
$\mathbb P\in\mathcal P_x(G)$, proving \eqref{eq:extraction_family}.
Moreover, \cite[Corollary~3.8]{park2026global} and
\eqref{eq:extraction_family} give $\mathcal U_x(G)\neq\varnothing$.

Finally, we prove uniqueness of the pair $(A,\mathbb Q)$ satisfying
$\Phi(A,\mathbb Q)=\mathbb P$, with cumulative discounting processes
identified up to $\mathbb Q$-indistinguishability.
Let $(A,\mathbb Q)$ be such a pair.
The definition of $\Phi$ gives
\eqref{eq:extraction_survival} with $A$ in place of $A^k$,
and $A_{\cdot\wedge\zeta}(\overline Y)$ is a continuous
compensator of $N$ under $\mathbb P$.
By uniqueness of compensators and
\eqref{eq:extraction_compensator},
$\mathbb P(A_{t\wedge\zeta}(\overline Y)=\Lambda_t
\text{ for all }t\ge0)=1$.
For bounded nonnegative $\widehat{\mathcal F}_t$-measurable $H$,
replace $H$ by $(e^{A_t}\wedge m)H$ in
\eqref{eq:extraction_survival} and let $m\uparrow\infty$.
Monotone convergence gives
\begin{equation}\label{eq:extraction_inverse_identity}
 \mathbb E^{\mathbb Q}[H(X)\mathbb I_{\{t<\tau_{\mathrm{exp}}\}}]
 =\mathbb E^{\mathbb P}
   [e^{\Lambda_t}H(\overline Y)\mathbb I_{\{t<\zeta\}}].
\end{equation}
Decomposing $H=H^+-H^-$ extends
\eqref{eq:extraction_inverse_identity} to every bounded
$\widehat{\mathcal F}_t$-measurable $H$.
Its right-hand side depends only on $\mathbb P$.
The finite-dimensional argument following
\eqref{eq:extraction_survival}, with $\tau_{\mathrm{exp}}$
in place of $\zeta$, therefore determines $\mathbb Q$ uniquely.
For this $\mathbb Q$, suppose that $A$ and $A'$ satisfy $\Phi(A,\mathbb Q)=\Phi(A',\mathbb Q)=\mathbb P$.
Uniqueness of compensators and
\eqref{eq:extraction_survival} imply
$\mathbb E^{\mathbb Q}
[e^{-A_t}\mathbb I_{\{A_t\ne A'_t\}}
\mathbb I_{\{t<\tau_{\mathrm{exp}}\}}]=0$
for every $t\ge0$.
Since $e^{-A_t}>0$ on $\{t<\tau_{\mathrm{exp}}\}$,
the two processes agree at every rational time before
$\tau_{\mathrm{exp}}$, outside a single $\mathbb Q$-null set.
Continuity and \eqref{eq:prop:pairDUS_to_DSVR_1} then give
$\mathbb Q(A_t=A'_t\text{ for all }t\ge0)=1$,
completing the proof of uniqueness.
\end{proof}

\begin{proof}[Proof of Theorem~\ref{thm:representation}]
We first show that $\mathcal T^{\mathcal U(G)}$ is a dynamic
sublinear valuation rule satisfying Assumption~\ref{ass:generator}
with generating function $G$.
By Lemma~\ref{lem:uncertainty_structure},
every $\mathcal U_x(G)$ is nonempty and its elements satisfy
\eqref{eq:prop:pairDUS_to_DSVR_1}.
Equations~\eqref{eq:extraction_family}
and~\eqref{eq:extraction_survival} identify
\eqref{eq:representation} with \eqref{eq:canonical_value},
so $\mathcal T^{\mathcal U(G)}=\mathcal V^G$ on $C_b(D)$.
By \cite[Corollary~3.10]{park2026global},
$\mathcal V^G$ satisfies conditions~\ref{ax:sublinear},
\ref{ax:contraction}, \ref{ax:order}, and~\ref{ax:semigroup}.
Under Assumption~\ref{ass:comparison}, \cite[Theorem~5.5]{park2026global} gives joint continuity of $(t,x)\mapsto\mathcal V_t^Gf(x)$ for every $f\in C_b(D)$.
Consequently, if $t_n\to t\ge0$, then
$\sup_{x\in K}|\mathcal V_{t_n}^Gf(x)-\mathcal V_t^Gf(x)|\to0$
for every compact $K\Subset D$, while
$\sup_n\|\mathcal V_{t_n}^Gf\|_\infty\le\|f\|_\infty$
by \ref{ax:contraction}.
Thus $\mathcal T^{\mathcal U(G)}=\mathcal V^G$
also satisfies \ref{ax:continuity}.
Moreover, continuity in \ref{item:G1} and
\cite[Corollary~5.2]{park2026global} give
\begin{equation}\label{eq:canonical_exact_generator}
\frac{\mathcal V_h^Gf-f}{h}
\longrightarrow G(\cdot,f,\nabla f,\nabla^2f)
\quad\text{locally uniformly on }D,\qquad
f\in C_b^\infty(D).
\end{equation}
Equation~\eqref{eq:canonical_exact_generator}
proves Assumption~\ref{ass:generator}
and identifies the generating function as $G$.

Finally, Proposition~\ref{prop:valuation_pde} shows that, for every
$f\in C_b(D)$, both $(t,x)\mapsto\mathcal T_t f(x)$ and
$(t,x)\mapsto\mathcal T_t^{\mathcal U(G)}f(x)$ are bounded viscosity
solutions of \eqref{eq:pde} with initial value $f$.
Assumption~\ref{ass:comparison} therefore identifies them,
proving \eqref{eq:valuation_identity}.
\end{proof}

\section*{Funding}
Hyungbin Park acknowledges support from National Research Foundation of Korea grants funded by the Ministry of Science and ICT (2022R1A5A6000840, RS-2026-25488333).
Financial support from the Institute for Research in Finance and Economics of Seoul National University is gratefully acknowledged. 
David Criens gratefully acknowledges support from the Freiburg Center for Data Analysis, Modeling and AI.

\section*{Conflict of interest}
The authors have no conflicts of interest to declare.

\section*{Author contributions}
All authors equally contributed to the writing and to the review of the entire
manuscript.

\section*{Data availability}
No datasets were generated or analysed during this study.

\section*{Declaration of AI assistance}
The authors used ChatGPT (OpenAI) to improve the language and readability of the article and to assist in correcting errors. The results presented in this article are the authors' original intellectual contributions. The authors take full responsibility for the content of the article, including its mathematical statements and proofs
 
% The bibliography is embedded so this source compiles without a separate .bib file.

\end{document}